\documentclass[12pt,a4paper]{article}
\usepackage[T1]{fontenc}
\usepackage{lmodern}
\usepackage[margin=2.5cm]{geometry}
\usepackage[onehalfspacing]{setspace}
\usepackage{amsmath,amssymb,amsthm}
\usepackage{graphicx}
\usepackage{booktabs}
\usepackage{xcolor}
\usepackage{csquotes}
\usepackage[authordate,backend=biber]{biblatex-chicago}
\usepackage{xurl}
\AtBeginBibliography{\Urlmuskip=0mu plus 0.5mu}
\usepackage[hidelinks]{hyperref}
\input{glyphtounicode}
\DefineBibliographyStrings{english}{andothers={et\addnbspace al\adddot}}
\defbibheading{bibliography}[\refname]{\section*{#1}}

\theoremstyle{plain}
\newtheorem{lemma}{Lemma}
\newtheorem{proposition}{Proposition}
\newtheorem{corollary}{Corollary}
\theoremstyle{definition}
\newtheorem{assumption}{Assumption}
\newtheorem{definition}{Definition}
\newtheorem{example}{Example}
\newtheorem{remark}{Remark}
\newcommand{\E}{\mathbb{E}}
\DeclareMathOperator{\vech}{vech}
\DeclareMathOperator{\Var}{Var}
\DeclareMathOperator{\Cov}{Cov}
\DeclareMathOperator{\diag}{diag}
\DeclareMathOperator{\rank}{rank}
\newcommand{\approvedstart}[1]{}
\newcommand{\approvedend}{}
\newcommand{\resumehere}[1]{}

\title{\Large Identification through Heteroskedasticity in One Lemma:\\[0.5ex] \large Partial Identification, an Eigenvalue Construction, and a Caution on Rank Tests}
\author{George Milunovich\thanks{Macquarie University, Sydney, NSW 2109, Australia. E-mail: george.milunovich@mq.edu.au.}}
\date{30 October 2026}

\begin{document}

\maketitle

\begin{abstract}
\noindent This teaching note unifies known results on identification through heteroskedasticity, under volatility regimes, GARCH, stochastic volatility and covariate-driven variances, in one lemma with a short linear-algebra proof. A shock is identified up to sign and labelling if and only if no other shock shares its variance path, and shocks that share a path are identified up to a rotation. The rows of the structural coefficient matrix solve one generalised eigenvalue problem that nests Rigobon's (2003) two-regime solution. Rank tests on second moments of the errors count at most linearly independent centred variance paths, not identified shocks. A simulation illustrates each point.

\medskip
\noindent\textbf{Keywords:} identification through heteroskedasticity, partial identification, structural vector autoregression, simultaneous equations, rank tests, teaching econometrics

\noindent\textbf{JEL classification:} C32, C30, A23
\end{abstract}

\thispagestyle{empty}
\newpage
\setcounter{page}{1}

\section{Introduction}\label{sec:1}

\approvedstart{23 October 2026}%
In simultaneous equation systems and structural vector autoregressions (SVARs), changes in the variances of the structural shocks can supply the missing identifying information when credible exclusion restrictions are unavailable and sign restrictions identify only a set of models. The idea goes back to \textcite{Wright:28}, and in modern form to \textcite{sentana2001identification} and \textcite{Rigobon:03}. It has since been applied with two or more volatility regimes \parencite{Lanne/lue:06,LLM2010}, with generalised autoregressive conditional heteroskedasticity (GARCH) \parencite{Lanne/Saikk:05,MilunovichYang2013,Luetkepohl/Milunovich:14}, with stochastic volatility \parencite{LSUW2026}, with nonparametric latent volatility \parencite{Lewis2021}, and with variances driven by observed covariates \parencite{Lewbel:06,milunovich2018simultaneous}. \textcite{Lewis2025} provides a survey of the literature.\approvedend

\approvedstart{23 October 2026}%
In an econometrics course, these cases can be presented separately, each with its own identification argument. Examples are an eigen-decomposition for two regimes and a rank condition on the variance paths for conditionally heteroskedastic factor models. Students may then find it hard to see how these different cases are connected. Both arguments, as usually stated, assume that every shock is identified. The eigen-decomposition requires distinct eigenvalues, and the rank condition requires linearly independent variance paths. When the heteroskedasticity is incomplete, because only some shocks are heteroskedastic or because some shocks' variances change in proportion, the answer to which parameters remain identified comes from separate results under different conditions. This partial case is important in practice. For example, \textcite{BBKM2026} find that a volatility break in October 1987 identifies only one of the three shocks of a standard oil-market model, because a test does not reject that the other two shocks' variances change in the same proportion. \textcite{sentana2001identification} show that such shocks remain identified when the other shocks share a variance path, and \textcite{LMY2020} when the other shocks are homoskedastic.\approvedend

\approvedstart{23 October 2026}%
We make two contributions to teaching this material. First, we collect these related cases in one lemma. Its proof uses nothing beyond the facts that an orthogonal matrix has orthonormal rows and that diagonal matrices can be compared element by element. The lemma states exactly what is identified, so its condition is both necessary and sufficient. Second, we derive two consequences that help in teaching, a constructive version based on one weighted moment and a caution about what rank tests can detect, both of which are illustrated by simulation. We make no claim of novelty for the lemma. Its argument is essentially that of Proposition~1 in \textcite{LLM2010}, for regimes, and of Theorem~A.1, with Lemma~A.1 and Corollary~A.1, in \textcite{LSUW2026}, for general variance sequences. The lemma collects in one statement the sufficiency direction, its converse, and the rotation within a class that shares a path \parencite[Proposition~4]{sentana2001identification}, under a normalisation that treats all cases alike.\approvedend

\approvedstart{29 October 2026}%
We summarise the results as follows. A column of the structural impact matrix $B$, which maps the structural shocks into the reduced-form errors, is identified up to sign and labelling if and only if no other shock's variance is proportional to that of the shock that multiplies it. Shocks whose variances are proportional are identified only as a class. The condition allows the variance paths to be linearly dependent, as they always are with two volatility regimes and two or more heteroskedastic shocks, so linear independence, which many published identification conditions require, is sufficient but not necessary. The rows of the structural coefficient matrix $A$, the inverse of $B$, solve a generalised eigenvalue problem built from a weighted second moment, which nests the two-regime solution of \textcite{Rigobon:03}. Rank tests of the type of \textcite{CraggDonald1997} and \textcite{KleibergenPaap2006}, applied to the covariance between the squares and cross-products of the errors and functions of observed covariates, count at most the number of linearly independent centred paths, which can be smaller or larger than the number of identified shocks. With two regimes such tests find a rank of one even when every shock is identified.\approvedend

\approvedstart{23 October 2026}%
We set out the model with a motivating example in Section~\ref{sec:2}, and state and prove the lemma and its corollaries in Section~\ref{sec:3}. The standard cases follow in Section~\ref{sec:5}. The eigenvalue construction and the caution on rank tests are in Section~\ref{sec:6}, the simulation exercise is in Section~\ref{sec:8}, and Section~\ref{sec:9} concludes.\approvedend

\section{Setting}\label{sec:2}

\approvedstart{23 October 2026}%
Consider a market in which demand is ${q=-\beta p+\varepsilon_d}$ and supply is ${q=\alpha p+\varepsilon_s}$, with uncorrelated shocks. Solving the two equations gives the reduced form
\begin{equation}\label{eq:reduced}
p=\frac{\varepsilon_d-\varepsilon_s}{\alpha+\beta}\quad\text{and}\quad q=\frac{\alpha\varepsilon_d+\beta\varepsilon_s}{\alpha+\beta},
\end{equation}
which expresses price and quantity as linear combinations of the two shocks. The data reveal only three moments, $\Var(p)$, $\Var(q)$ and $\Cov(p,q)$. These must determine four unknowns, the slopes $\alpha$ and $\beta$ and the variances $\Var(\varepsilon_d)$ and $\Var(\varepsilon_s)$, so many pairs of supply and demand curves fit the data equally well. Now suppose that the shock variances differ between a calm and a volatile regime, and that the regime is observed. Within each regime $s=1,2$ the data reveal $\Var_s(p)$, $\Var_s(q)$ and $\Cov_s(p,q)$, so the two regimes give six moments. The unknowns are the two slopes, which do not change, and the two shock variances in each regime, also six in all, so the two regimes can identify the system \parencite{Wright:28,Rigobon:03}. This count is the order condition, which is necessary but not sufficient for identification. In particular, the six moments identify the system if and only if the two shocks' variances change in different proportions between the regimes. This is the rank condition in this example. For example, suppose that only the variance of the demand shock is larger in the volatile regime. The demand curve then shifts more while supply stays put, and the scatter of $(p,q)$ stretches along the supply curve and traces it out. Section~\ref{sec:3} shows that, in general, the rank condition is that no two shocks share a variance path.\approvedend

\begin{figure}[!ht]
\centering
\includegraphics[width=\linewidth]{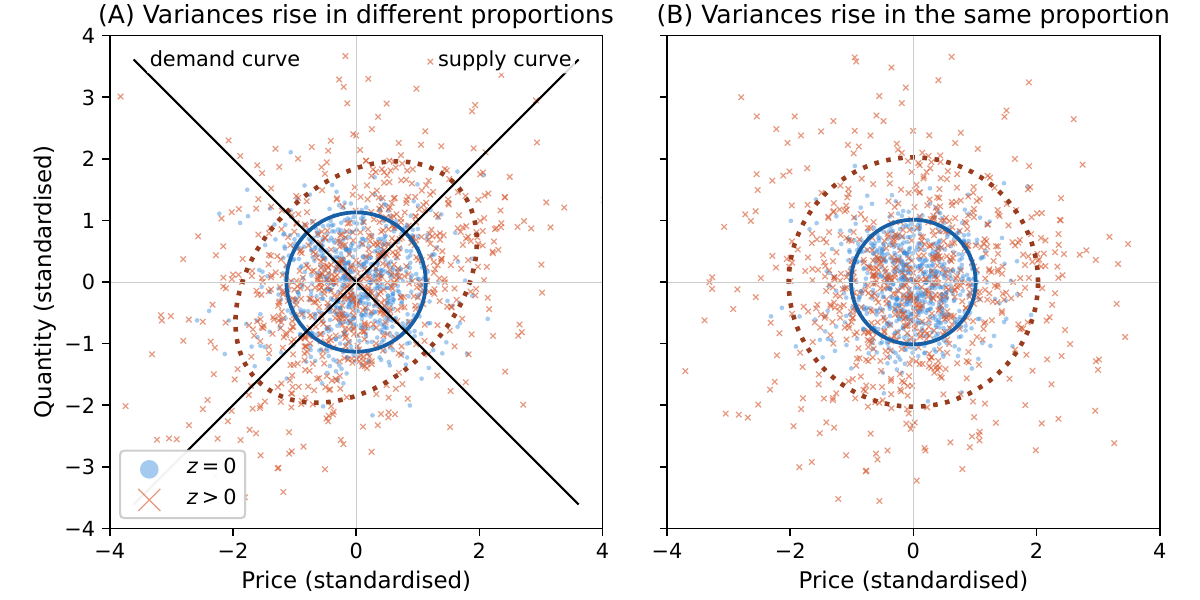}
\caption{Simulated price and quantity in the demand and supply example.}\label{fig:ellipses}

\medskip
\parbox{\linewidth}{\footnotesize Notes: Demand is $q=-p+\varepsilon_d$ and supply is $q=p+\varepsilon_s$, and $z$ is an observed variable that drives the variances of the shocks. Panel A: $\Var(\varepsilon_d\mid z)=1+3z$ and $\Var(\varepsilon_s\mid z)=1+z$. Panel B: $\Var(\varepsilon_d\mid z)=\Var(\varepsilon_s\mid z)=1+3z$. Each of the 1,500 observations has $z=0$ with probability one half (circles, solid ellipse) and otherwise $z$ drawn uniformly on $(0,2)$ (crosses, dotted ellipse). Price and quantity are in deviations from their means, divided by their standard deviations. In Panel A, straight lines are the demand and supply curves. Ellipses are contours of the covariance matrix of each group implied by the model.}
\end{figure}

\approvedstart{23 October 2026}%
Figure~\ref{fig:ellipses} shows the same effect with an observed driver $z$ in place of the regime indicator. When $z=0$ the two shocks have the same variance, and when $z>0$ the variances rise with $z$. In Panel A both variances rise, but the variance of the demand shock rises faster. The observations with $z=0$ then form a circle, while those with $z>0$ form a larger ellipse stretched along the supply curve. The demand curve is not traced out, but it is identified too. Along it, the spread of the observations grows less than along the supply curve. The circle and the ellipse have different shapes, and only the true pair of curves fits both, so demand and supply are both identified. In Panel B both variances rise in the same proportion. The observations with $z>0$ then form a larger circle, so the second group adds no information about the slopes, and the shocks are not identified. The rest of the note generalises these two observations.\approvedend

\approvedstart{23 October 2026}%
Let $u_i$ be a $K\times1$ vector of reduced-form errors, for example the residuals of a vector autoregression (VAR) or of the reduced form of a simultaneous equation system, such as the demand and supply model above, and suppose that
\begin{equation}\label{eq:model}
u_i=B\varepsilon_i,\quad\text{or equivalently}\quad Au_i=\varepsilon_i\quad\text{with}\quad A=B^{-1},
\end{equation}
with $\varepsilon_i$ the $K\times1$ vector of structural shocks, and $B$ the nonsingular $K\times K$ structural impact matrix, whose columns are $b_1,\dots,b_K$. The index $i$ runs over time periods or cross-sectional units. The conditioning information $W_i$ is the past in a time series model, a regime indicator, or observed covariates. In stochastic volatility models it also includes the unobserved volatility processes.\approvedend

\begin{assumption}\label{ass:cond}
\approvedstart{23 October 2026}%
The structural shocks satisfy $\E(\varepsilon_i\mid W_i)=0$ and $\Var(\varepsilon_i\mid W_i)=H_i=\diag(h_{1i},\dots,h_{Ki})$, with $h_{ki}>0$.
\end{assumption}

\begin{assumption}\label{ass:norm}
The conditional variances satisfy $\E(H_i)=I_K$, where the expectation is over the distribution of $W_i$ (or, for a deterministic index such as regimes, over the relative frequencies of the regimes).
\end{assumption}

Assumption~\ref{ass:cond} says that the structural shocks are conditionally uncorrelated. Assumption~\ref{ass:norm} fixes the scale of each shock. It is a normalisation, not a restriction. Without it, $(BD,\,D^{-1}H_iD^{-1})$ would fit the data as well as $(B,H_i)$ for any positive diagonal $D$. \textcite{LLM2010} and \textcite{LSUW2026} instead normalise the variances in the first regime or period. Normalising the mean treats regimes, time series and covariates alike. Under both assumptions, the conditional and unconditional covariance matrices of the errors are
\begin{equation}\label{eq:omega}
\Omega_i:=\Var(u_i\mid W_i)=BH_iB'\quad\text{and}\quad\Omega:=\E(u_iu_i')=BB'.
\end{equation}
In the time-series, regime and covariate cases below, $\Omega_i$ is written $\Omega_t$, $\Omega_s$ and $\Omega(z)$, respectively.\approvedend

\begin{definition}
\approvedstart{24 October 2026}%
Let $H$ denote the collection of the conditional variance matrices $H_i$ over all $i$. A pair $(\tilde B,\tilde H)$ satisfying Assumptions~\ref{ass:cond} and~\ref{ass:norm} is \emph{observationally equivalent} to $(B,H)$ if $\tilde B\tilde H_i\tilde B'=BH_iB'$ for all $i$ (almost surely).
\end{definition}

\begin{definition}
The \emph{variance path} of shock $k$ is the collection $h_k$ of its conditional variances $h_{ki}$ over all $i$. Shocks $j$ and $k$ \emph{share a path} if $h_{ji}=h_{ki}$ for all $i$ (almost surely). Sharing a path sorts the $K$ shocks into classes, so that each shock belongs to exactly one class, and $B_C$ denotes the columns of $B$ for the shocks in class $C$.
\end{definition}

For example, if shocks 2 and 3 share a path and shock 1 has a path of its own, then $\{1\}$ and $\{2,3\}$ are classes, and $B_{\{2,3\}}=(b_2,b_3)$. Because of Assumption~\ref{ass:norm}, all homoskedastic shocks have the path $h_k\equiv1$, so they form one class. Two heteroskedastic shocks whose variances are proportional also share a path, since the normalisation removes the factor of proportionality.\approvedend

\approvedstart{24 October 2026}%
The \emph{statistical step} of identification asks which conditional covariance matrices $\Omega_i$ of the reduced-form errors the data reveal. The answer depends on the model. With regimes, $\Omega_s$ is simply the covariance matrix of the errors within regime $s$. With stochastic volatility, $\Omega_t$ depends on unobserved volatilities, so whether the data reveal it is itself a question. Section~\ref{sec:5} discusses each case. The \emph{algebraic step}, the identification problem in the classical sense, asks which $B$ are consistent with given $\Omega_i$. Separating the two steps lets the lemma address only the algebraic step, which is common to all cases.\approvedend

\section{The lemma and its corollaries}\label{sec:3}

\begin{lemma}\label{lem:main}
\approvedstart{24 October 2026}%
Let Assumptions~\ref{ass:cond} and~\ref{ass:norm} hold, let $(\tilde B,\tilde H)$ be observationally equivalent to $(B,H)$, and let $\Gamma=\tilde B^{-1}B$, with elements $\gamma_{kj}$. Then:
\begin{enumerate}
\item[(a)] The matrix $\Gamma$ is orthogonal, and $\tilde B=B\Gamma'$.
\item[(b)] If $\gamma_{kj}\neq0$, then $\tilde h_k=h_j$.
\item[(c)] For every class $C$ with $m$ members, exactly $m$ shocks of $(\tilde B,\tilde H)$ have the variance path of $C$, and the corresponding columns of $\tilde B$ are $B_CQ_C$ for some $m\times m$ orthogonal matrix $Q_C$.
\end{enumerate}
Conversely, for any orthogonal matrices $Q_C$, one for each class, and any relabelling of the shocks, the resulting $\tilde B$, with the variance paths relabelled accordingly, is observationally equivalent to $(B,H)$.\approvedend
\end{lemma}

\begin{proof}
\approvedstart{24 October 2026}%
(a) Observational equivalence gives $\tilde H_i=\Gamma H_i\Gamma'$ for all $i$. Taking expectations and using Assumption~\ref{ass:norm} for both structures gives $I_K=\Gamma\Gamma'$. Hence $\Gamma$ is orthogonal, and $\tilde B=B\Gamma^{-1}=B\Gamma'$.

(b) Multiplying $\tilde H_i=\Gamma H_i\Gamma'$ on the right by $\Gamma$ gives $\Gamma H_i=\tilde H_i\Gamma$. Both $H_i$ and $\tilde H_i$ are diagonal, so the $(k,j)$ element reads $\gamma_{kj}h_{ji}=\tilde h_{ki}\gamma_{kj}$ for all $i$. If $\gamma_{kj}\neq0$, then $h_{ji}=\tilde h_{ki}$ for all $i$ (almost surely).

(c) By (b), all nonzero elements of row $k$ of $\Gamma$ lie in columns whose shocks have the path $\tilde h_k$, so each row is supported on a single class (every row is nonzero, since $\Gamma$ is orthogonal). The rows supported on a class $C$ are orthonormal vectors in the $m$-dimensional coordinate space of $C$, so there are at most $m$ of them. Since there are $K$ rows in total and the class sizes add up to $K$, there are exactly $m$. Restricted to the columns in $C$, these rows form an $m\times m$ orthogonal matrix, $Q_C'$ say. By (a), the corresponding columns of $\tilde B$ are $B_CQ_C$, and the corresponding shocks have the variance path of $C$.

For the converse, let $\Gamma$ be block diagonal with blocks $Q_C'$, followed by a permutation. On the coordinates of a class, $H_i$ is a scalar multiple of the identity, $\phi_iI_m$ say, so $Q_C'(\phi_iI_m)Q_C=\phi_iI_m$. Hence $\Gamma H_i\Gamma'$ is diagonal, equal to $H_i$ up to the relabelling, and satisfies Assumption~\ref{ass:norm}.\approvedend
\end{proof}

\approvedstart{24 October 2026}%
In words, \emph{a column of $B$ is identified up to sign and labelling if and only if the variance path of the corresponding shock, the shock that multiplies that column, differs from those of all other shocks.} A class of $m\ge2$ shocks with a common path is identified only up to an $m\times m$ rotation of its columns. Their column space and their contribution $B_CB_C'$ to $\Omega$ are still identified. Distinctness of the paths is therefore necessary and sufficient. Necessity relies on $B$ being otherwise unrestricted and on using only the second moments of the errors. Restrictions can identify a rotation within a class that the variances leave free, as zero restrictions on $B$ do in \textcite{BBKM2026}, and as a unit diagonal on $A$ combined with variance functions without intercepts does locally, that is, uniquely in a neighbourhood of the true parameters, in \textcite[Example~1]{milunovich2018simultaneous}. In addition, higher moments can identify more. When the shocks are mutually independent and at most one of them is Gaussian, $B$ is identified up to sign and labelling even without heteroskedasticity \parencite{LanneMeitzSaikkonen2017}. The labels of the shocks are never identified by the second moments alone, and any labelling rule, such as ordering the shocks by the size of a variance parameter, can fix them.\approvedend

\begin{corollary}[Full identification]\label{cor:full}
\approvedstart{24 October 2026}%
If no two shocks share a variance path, then $B$ is identified up to the labelling and signs of its columns, and the variance paths are identified up to the same labelling. In particular, the condition allows at most one homoskedastic shock.\approvedend
\end{corollary}

\begin{proof}
\approvedstart{24 October 2026}%
If no two shocks share a path, every class has one member. By Lemma~\ref{lem:main}(c), each column of $\tilde B$ is then $b_kQ_C$, where $Q_C$ is a $1\times1$ orthogonal matrix. The only $1\times1$ orthogonal matrices are $+1$ and $-1$, so each column of $\tilde B$ is $\pm b_k$ for some $k$, and its variance path is that of shock $k$. Hence $B$ is identified up to the labelling and signs of its columns, and the paths up to the same labelling. At most one shock can be homoskedastic, because under Assumption~\ref{ass:norm} all homoskedastic shocks share the constant path.\approvedend
\end{proof}

\begin{corollary}[Partial identification]\label{cor:partial}
\approvedstart{24 October 2026}%
Suppose that shocks $1,\dots,r$ have variance paths that differ from one another and from the constant path, and that shocks $r+1,\dots,K$ are homoskedastic. Partition $B=[B_1,B_2]$ after the first $r$ columns, and let $A_1$ be the first $r$ rows of $A$. Then $B_1$ is identified up to replacement by $B_1PD$, where $P$ is a permutation matrix and $D$ a diagonal matrix of signs. The block $B_2$ is identified only up to replacement by $B_2Q$ with $Q$ orthogonal, and $A_1=B_1'\Omega^{-1}$ is identified up to the labelling and signs of its rows. If $r=K-1$, $B$ is identified up to labelling and sign.\approvedend
\end{corollary}

\begin{proof}
\approvedstart{24 October 2026}%
Each of shocks $1,\dots,r$ forms a class of its own, and shocks $r+1,\dots,K$ form one class. By Lemma~\ref{lem:main}(c), the columns of $\tilde B$ for the first $r$ classes are $\pm b_k$, which gives $B_1PD$, and the columns for the homoskedastic class are $B_2Q$ with $Q$ orthogonal. For $A_1$, note that $\Omega=BB'$ implies $A=B^{-1}=B'(BB')^{-1}=B'\Omega^{-1}$, so $A_1=B_1'\Omega^{-1}$, and $\Omega$ is identified. If $r=K-1$, the homoskedastic class has one member, and $Q=\pm1$.\approvedend
\end{proof}

\approvedstart{24 October 2026}%
Corollary~\ref{cor:full} is the familiar full-identification result \parencite{sentana2001identification,Rigobon:03,milunovich2018simultaneous}, and Corollary~\ref{cor:partial} is the homoskedastic case of Proposition~4 of \textcite{sentana2001identification}, which allows the remaining shocks to share any common path, and the subvector result of \textcite{LMY2020}. Both subvector results assume that the variance paths are linearly independent. For Sentana and Fiorentini this means the heteroskedastic paths together with the common path of the remaining shocks. Remark~\ref{rem:1} shows that this is more than the corollaries need.\approvedend

\approvedstart{24 October 2026}%
To see Corollary~\ref{cor:partial} in terms of counts, take $K=3$, one heteroskedastic shock and two regimes, with the heteroskedastic shock having variance $\lambda_{11}$ in regime~1 and $\lambda_{12}$ in regime~2. The two regime covariance matrices give twelve moments, the variances of the three errors and their three covariances in each regime. That seems enough for the ten unknowns, the nine elements of $B$ and one variance parameter of the heteroskedastic shock. Its two regime variances, $\lambda_{11}$ and $\lambda_{12}$, are tied together by Assumption~\ref{ass:norm}, so only one of them is free. However, $\Omega_2-\Omega_1=(\lambda_{12}-\lambda_{11})\,b_1b_1'$ has rank one, so it contains only three independent numbers. The two regimes therefore give nine independent moments instead of twelve, for ten unknowns. The missing unknown is the angle of a rotation between the two homoskedastic shocks, which the data cannot determine. In general, one such angle remains for each pair of homoskedastic shocks, $(K-r)(K-r-1)/2$ in all, so none remain when only one shock is homoskedastic.\approvedend

\begin{remark}\label{rem:1}
\approvedstart{24 October 2026}%
Many published identification conditions require the variance paths to be linearly independent, in the sense that $1$, $h_1,\dots,h_r$ are linearly independent. That condition is sufficient but not necessary. With two regimes, every variance path is a linear combination of a constant and the regime indicator, so linear independence fails whenever $r\ge2$. Yet the lemma shows that all shocks with distinct paths are identified. \textcite[Theorem~1, Corollary~1 and Example~1]{milunovich2018simultaneous} show that linear independence is not necessary for identification by quasi-maximum likelihood. \textcite[Theorem~1 and Corollary~1]{Lewis2021} likewise identifies a shock when an autocovariance matrix of the squares and cross-products of the errors has rank at least two and no other row is proportional to the shock's own row, a condition weaker than full rank.\approvedend
\end{remark}

\begin{remark}\label{rem:2}
\approvedstart{24 October 2026}%
If the variance paths belong to a common parametric family, $h_{ki}=G(W_i,\theta_k)$, and $\theta\mapsto G(\cdot,\theta)$ is injective on the support of $W_i$, the lemma also identifies $(\theta_1,\dots,\theta_K)$ up to the same labelling. If the $h_k$ are unknown functions of observed covariates, it identifies the functions themselves (Example~\ref{ex:cov}).\approvedend
\end{remark}

\section{Worked examples}\label{sec:5}

\approvedstart{24 October 2026}%
The examples differ only in the index $i$, the conditioning information $W_i$, and the statistical step that delivers $\Omega_i$. Counting moments and unknowns is a useful first exercise in each case. For a system of $K$ variables, the data supply $K(K+1)/2$ second moments for each distinct $\Omega_i$, the $K$ conditional variances of the errors and their $K(K-1)/2$ conditional covariances. The unknowns are the $K^2$ elements of $B$ and the parameters of the variance paths, less the $K$ scale normalisations of Assumption~\ref{ass:norm}. In terms of the structural equations, the $K^2$ elements of $B$ are in one-to-one correspondence with the $K(K-1)$ structural coefficients, such as the slopes $\alpha$ and $\beta$ in Section~\ref{sec:2}, and the $K$ scales of the shocks. The correspondence holds up to the sign of each row of $A$, provided every diagonal element of $A$ is nonzero. Table~\ref{tab:count} summarises the counts. They check only the order condition. The rank condition is always the lemma's condition that the paths be distinct. The last row shows that the moments must be counted after removing redundancies.\approvedend

\begin{table}[ht]
\centering
\caption{Moments and unknowns in the standard cases.}\label{tab:count}
\small
\begin{tabular}{lccc}
\toprule
Case & Moments & Unknowns & $K=2$\\
\midrule
No heteroskedasticity & $K(K+1)/2$ & $K^2$ & 3 vs 4\\
Two regimes & $K(K+1)$ & $K^2+K$ & 6 vs 6\\
GARCH(1,1), $T$ periods & $TK(K+1)/2$ & $K^2+2r$ & $3T$ vs 8\\
Covariate, $n_z$ values, $h_k\propto e^{\delta_kz}$ & $n_zK(K+1)/2$ & $K^2+K$ & $3n_z$ vs 6\\
Unknown $h_k(z)$, $n_z$ values$^{a}$ & $n_zK(K-1)/2$ & $K(K-1)$ & $n_z$ vs 2\\
Stochastic volatility, $T$ periods$^{b}$ & $TK(K+1)/2$ & $K(K-1)+TK$ & $3T$ vs $2+2T$\\
\midrule
$K=3$, $r=1$, two regimes & 9 independent & 10 & --\\
\bottomrule
\end{tabular}
\par\vspace{1ex}\parbox{\linewidth}{\footnotesize Notes: Section~\ref{sec:5} explains the counts, and Section~\ref{sec:3} the last row. $T$ is the number of periods. In the column for $K=2$, all shocks are heteroskedastic except in the first row. $^{a}$Off-diagonal moments only, with the $K(K-1)$ structural coefficients as unknowns (Example~\ref{ex:cov}). $^{b}$If $\Omega_t$ were observed, with the latent variances as unknowns, excluding the parameters of the latent processes.}
\end{table}

\begin{example}[Two volatility regimes]\label{ex:regimes}
\approvedstart{24 October 2026}%
Let $i=t$, with a known regime $s_t\in\{1,2\}$ occurring with frequencies $\pi$ and $1-\pi$. Shock $k$ has variance $\lambda_{k1}$ in regime~1 and $\lambda_{k2}$ in regime~2, with $\pi\lambda_{k1}+(1-\pi)\lambda_{k2}=1$. The normalisation pins down the level, so two paths coincide if and only if the shocks have the same ratio $\lambda_{k2}/\lambda_{k1}$. Corollary~\ref{cor:full} is therefore the classical condition of distinct variance ratios \parencite{Rigobon:03,LLM2010}, and a class of shocks with a common ratio is identified only up to rotation \parencite{Lewis2025,BBKM2026}. The statistical step is immediate, because $\Omega_s=\E(u_tu_t'\mid s_t=s)$ is a regime covariance matrix. The two matrices give $K(K+1)$ moments for $K^2+K$ unknowns (Table~\ref{tab:count}), so the order condition holds exactly, with no overidentifying moments.\approvedend
\end{example}

\begin{example}[GARCH-type conditional heteroskedasticity]
\approvedstart{24 October 2026}%
Let $i=t$, and let $W_t$ be the information available at the end of period $t-1$, which includes the past errors $u_{t-1},u_{t-2},\dots$. If $r$ shocks follow univariate GARCH processes and $K-r$ are homoskedastic, the variance path of each GARCH shock is the sequence of its conditional variances over time. Two GARCH shocks with identical parameters still have different paths almost surely, provided their variances respond to the lagged squared shocks, because each variance is driven by its own past shocks. The distinctness condition therefore concerns the paths, not the parameters. With $r\ge K-1$ GARCH shocks, Corollary~\ref{cor:full} gives identification, and \textcite{MilunovichYang2013} show that $B$ and the GARCH parameters are then jointly locally identified. With fewer than $K-1$ GARCH shocks, only the columns of the GARCH shocks are identified \parencite{Lanne/Saikk:05}. The statistical step uses the parametric model, in which $\Omega_t$ is a known function of the past and the parameters, estimated by quasi-maximum likelihood. Each period contributes $K(K+1)/2$ moments, while the unknowns do not grow with $T$ (Table~\ref{tab:count}), so the order condition is met many times over. Since each period supplies a single $u_t$ rather than $\Omega_t$ itself, the moments are used through the likelihood.\approvedend
\end{example}

\begin{example}[Stochastic volatility]
\approvedstart{27 October 2026}%
Let $i=t$, and let $W_t$ contain the information available at the end of period $t-1$ and the latent volatilities. Let $h_{kt}\propto\exp(\omega_k x_{kt})$ with latent autoregressive $x_{kt}$. The paths are distinct almost surely when at most one $\omega_k$ is zero and the latent processes are independent. \textcite{LSUW2026} use this specification and allow several $\omega_k$ to be zero, in which case Corollary~\ref{cor:partial} applies. These conditions are sufficient, but not necessary, for distinct paths. The statistical step is harder in this case, because $\Omega_t$ conditional on the latent volatilities is not a moment of observables. It is handled either through the likelihood of the full model or through the autocovariances of $\vech(u_tu_t')$, where $\vech$ stacks the elements of a symmetric matrix on and below the diagonal, as in \textcite{Lewis2021}.\approvedend
\end{example}

\begin{example}[Variances driven by observed covariates]\label{ex:cov}
\approvedstart{27 October 2026}%
Let $i$ index cross-sectional units and $h_{ki}=h_k(z_i)$ for observed covariates $z_i$, which are part of $W_i$ \parencite{Lewbel:06,milunovich2018simultaneous,LMY2020}. The statistical step is the conditional moment $\Omega(z)=\E(uu'\mid z)$, and Lemma~\ref{lem:main} gives identification without any parametric model for the variances. The shocks whose paths differ from all others are identified up to sign and labelling, and so are their variance functions, since $h_k(z)=a_k'\Omega(z)a_k$, where $a_k'$ is the $k$th row of $A$. In terms of the rows of $A$, the $K(K+1)/2$ moments at each value of $z$ split into $K(K-1)/2$ off-diagonal equations, $\E[(a_j'u)(a_k'u)\mid z]=0$ for $j\neq k$, and $K$ diagonal ones, $\E[(a_k'u)^2\mid z]=h_k(z)$. The diagonal equations determine the variance functions and, with Assumption~\ref{ass:norm}, the scales of the rows of $A$. The off-diagonal equations do not involve the variance functions and are unchanged when a row of $A$ is rescaled, so they must determine the $K(K-1)$ ratios within the rows of $A$, which correspond to the structural coefficients. With $K=2$ there is one off-diagonal moment for each value of $z$ and two unknowns, so one value is the unidentified case, two values are the regime case, and a continuous covariate gives infinitely many. For example, let $K=2$ and $h_k(z)=\theta_{k0}+\theta_{k1}z$, positive for every $z$. The two normalised paths coincide if and only if $(\theta_{10},\theta_{11})$ is proportional to $(\theta_{20},\theta_{21})$, so both shocks are identified unless the coefficient vectors are proportional. Figure~\ref{fig:ellipses} is of this form, with coefficient vectors $(1,3)$ and $(1,1)$ in Panel A and $(1,3)$ for both shocks in Panel B. For exponential functions, $h_k(z)\propto e^{\delta_kz}$, the condition is $\delta_1\neq\delta_2$ \parencite{milunovich2018simultaneous}.\approvedend
\end{example}

\section{An eigenvalue construction and a caution on rank tests}\label{sec:6}

\approvedstart{27 October 2026}%
Lemma~\ref{lem:main} states what is identified but not how to compute it. This section shows that the rows of $A$ solve a generalised eigenvalue problem built from a single weighted second moment of the errors, and then asks what a rank test on such moments counts. Constructions of this type appear in signal processing \parencite{Cardoso1989,Tong1991} and in econometrics \parencite{LMY2020}. The construction nests the two-regime solution of \textcite{Rigobon:03}, and the simulation in Section~\ref{sec:8} uses it.\approvedend

\begin{proposition}\label{prop:eigen}
\approvedstart{27 October 2026}%
Let $w(\cdot)$ be a function of the observed conditioning information, write $w_i$ for its value for observation~$i$, and suppose that $\E(|w_i|\,\|u_i\|^2)<\infty$. Let $M_w=\E[w_i\,u_iu_i']$ and $\mu_k=\E[w_ih_{ki}]$. Then
\begin{equation}\label{eq:eigen}
M_w\,a_k=\mu_k\,\Omega\,a_k,\qquad k=1,\dots,K,
\end{equation}
where $a_k'$ is the $k$th row of $A$. If $\mu_k$ differs from every $\mu_j$ with $j\neq k$, then $a_k$ is identified, up to sign and labelling, as the generalised eigenvector of $(M_w,\Omega)$ with eigenvalue $\mu_k$, normalised by $a_k'\Omega a_k=1$.\approvedend
\end{proposition}

\begin{proof}
\approvedstart{27 October 2026}%
By definition, $M_w=B\,\E(w_iH_i)\,B'$ and $\Omega=BB'$. Since $B'A'=I$, $M_wA'=B\,\E(w_iH_i)=\Omega A'\,\E(w_iH_i)$, and the $k$th column is the claim. The normalisation holds because $A\Omega A'=I$.\approvedend
\end{proof}

\approvedstart{27 October 2026}%
With $w_i$ the indicator of regime~2, $M_w/(1-\pi)$ is the regime-2 covariance matrix, and Proposition~\ref{prop:eigen} gives the just-identified two-regime solution of \textcite{Rigobon:03}, equivalent to the decomposition of the two regime covariance matrices used by \textcite{Lanne/lue:06} and \textcite{LLM2010}. Any weight will do, provided it gives the shocks distinct eigenvalues. When the paths are functions of observed information, as with regimes, covariates and GARCH, the weights that fail to separate two distinct paths form a proper linear subspace, so every weight outside a finite union of such subspaces separates all classes. A natural choice can still fail. In Design~A of Section~\ref{sec:8}, $h_2$ is an even function of a symmetric covariate $z$, so the weight $w=z$ gives $\E[z\,(h_2(z)-1)]=0$. Then $\mu_2$ equals the eigenvalue of the homoskedastic shocks, and the second shock cannot be separated from them, whereas $w=z+z^2/2$ separates it. Finally, for a weight with $\E[w_i(h_{ki}-1)]\neq0$ for every heteroskedastic shock $k$, the number of eigenvalues equal to $\E(w_i)$ is the number of homoskedastic shocks, which connects to tests of equal eigenvalues \parencite{LMNS2021}.\approvedend

\approvedstart{29 October 2026}%
Proposition~\ref{prop:eigen} suggests a direct way to count the identified shocks, by testing which eigenvalues are equal. In practice a rank test is often used instead, for a number of reasons. First, the sufficient conditions in much of the literature are rank conditions, such as linear independence of the variance paths \parencite{LMY2020}. \textcite{Lewis2021}, for instance, checks his identification conditions by testing whether an autocovariance matrix of the squares and cross-products of the errors has full rank. Second, rank tests with chi-squared limits apply to any matrix with an asymptotically normal estimator \parencite{CraggDonald1997,KleibergenPaap2006}. Tests of equal eigenvalues are harder to construct, because under the null the eigenvectors are not unique, and are available for two regimes with a known break date \parencite{LMNS2021}. Third, a single weight can fail to separate shocks with distinct paths, as $w=z$ does in Design~A, whereas a rank test uses many functions of the covariates at once. It is therefore tempting to count the heteroskedastic shocks by the rank of the covariance between the squares and cross-products of the errors and functions of the covariates. Proposition~\ref{prop:rank} shows that this rank measures something else.\approvedend

\begin{proposition}\label{prop:rank}
\approvedstart{27 October 2026}%
Let the variances be functions of covariates, as in Example~\ref{ex:cov}, so that $h_{ki}=h_k(z_i)$. Write $h(z)=(h_1(z),\dots,h_K(z))'$, suppose that $\E[h_k(z_i)^2]<\infty$ for every $k$, and let $\psi(z)$ be a vector of square-integrable functions. Then
\begin{equation}\label{eq:cpsi}
\begin{aligned}
C_\psi&:=\Cov\bigl(\vech(u_iu_i'),\psi(z_i)\bigr)=V\Cov\bigl(h(z_i),\psi(z_i)\bigr),\\
V&=\bigl[\vech(b_1b_1'),\dots,\vech(b_Kb_K')\bigr],
\end{aligned}
\end{equation}
where $V$ has full column rank $K$. Hence $\rank C_\psi\le r^*:=\rank \Var(h(z_i))$, with equality when $\psi$ contains $h$. If shocks $1,\dots,r$ are heteroskedastic and the rest homoskedastic, then $r^*\le r$, with equality if and only if $1,h_1,\dots,h_r$ are linearly independent.\approvedend
\end{proposition}

\approvedstart{29 October 2026}%
In words, $r^*$ is the number of linearly independent centred variance paths, that is, the number of independent directions in which $\Omega(z)$ moves as $z$ changes.\approvedend

\begin{proof}
\approvedstart{27 October 2026}%
The conditional mean is
\[
\E[\vech(u_iu_i')\mid z_i]=\vech(BH(z_i)B')=\textstyle\sum_kh_k(z_i)\vech(b_kb_k'),\]
which gives the factorisation. If $\sum_kc_k\,b_kb_k'=B\,\diag(c)\,B'=0$, then $c=0$ because $B$ is nonsingular, so $V$ has full column rank. The rank of $\Cov(h,\psi)$ is at most the dimension of the span of the centred $h_k$, which is $r^*$, and equals it when $\psi$ contains $h$. The centred paths of homoskedastic shocks are zero, so $r^*\le r$, and the centred $h_1,\dots,h_r$ are linearly independent exactly when $1,h_1,\dots,h_r$ are.\approvedend
\end{proof}

\approvedstart{29 October 2026}%
With two regimes, and treating the indicator $d$ of regime~2 as the covariate, $h_k=1+\rho_k(d-\E d)$ for constants $\rho_k$, which are zero for homoskedastic shocks. Every path is therefore affine in the single indicator $d$, so $\Omega_s$ moves along a single direction and $r^*\le1$ whatever the number of heteroskedastic shocks. A rank test on a matrix such as $C_\psi$ estimates $\rank C_\psi$, which by Proposition~\ref{prop:rank} is at most $r^*$. Read as a count of identified shocks, the estimated rank can therefore be too low, as in Design~B of Section~\ref{sec:8}, or too high, as in Design~C, where the rank is one but no shock is identified individually. A full-rank check, as in \textcite{Lewis2021}, does not have the second problem, because in his setting full rank implies identification. As he notes, the check can only be conservative, the point that Remark~\ref{rem:1} highlights. \textcite{Lewis2025} identifies factor structures in volatilities as a gap for existing tests of the number of heteroskedastic shocks.\approvedend

\approvedstart{29 October 2026}%
The number of identified shocks and the number $r$ of heteroskedastic shocks answer different questions. For a weight that separates the classes, a shock is identified when its eigenvalue in Proposition~\ref{prop:eigen} is distinct, so the number of identified shocks is a question about repeated eigenvalues. The number $r$ is instead tested sequentially, by asking whether any heteroskedasticity remains among the shocks not yet extracted. These tests assume nondegenerate GARCH components \parencite{Lanne/Saikk:05,Luetkepohl/Milunovich:14} or linearly independent paths \parencite{LMY2020}, under which every heteroskedastic shock is identified. Outside these assumptions fewer than $r$ shocks can be identified, as in Design~C of Section~\ref{sec:8}, where two shocks are heteroskedastic but share a path.\approvedend

\section{A simulation exercise}\label{sec:8}

\approvedstart{27 October 2026}%
To illustrate Lemma~\ref{lem:main} and Propositions~\ref{prop:eigen} and~\ref{prop:rank} in finite samples, we simulate three designs with $K=4$ from
\begin{equation}\label{eq:simmodel}
u_i=B\varepsilon_i,\qquad \varepsilon_i\mid z_i\sim N\bigl(0,\diag(h_1(z_i),h_2(z_i),1,1)\bigr),\qquad i=1,\dots,n,
\end{equation}
independently over $i$, where
\[
B=\begin{pmatrix}0.8&0.5&0.3&0.2\\0.4&1.2&0.5&0.1\\0.2&0.3&1.5&0.4\\0.3&0.2&0.4&0.9\end{pmatrix}
\]
and $z_i\sim N(0,1)$ is a scalar covariate. Shocks 1 and 2 are heteroskedastic, and shocks 3 and 4 are homoskedastic. In the two-regime designs, $d_i=1\{z_i>0\}$, which equals one when $z_i>0$ and zero otherwise, indicates regime~2. The designs differ only in $h_1$ and $h_2$.
\begin{itemize}
\item[A.] Covariate-driven variances, $h_1\propto e^{0.8z}$ and $h_2\propto0.2+z^2$, so $r=r^*=2$. The two paths are distinct, so both heteroskedastic shocks are identified individually.
\item[B.] Two regimes with distinct ratios, $h_1\propto1+2d$ and $h_2\propto1+0.4d$, so $r=2$ but $r^*=1$. The two paths are still distinct, so both heteroskedastic shocks are identified individually.
\item[C.] Two regimes with equal ratios, $h_1=h_2\propto1+2d$, so $r=2$ and $r^*=1$. The two heteroskedastic shocks share a path and form one class, so they are not identified individually.\approvedend
\end{itemize}

\approvedstart{30 October 2026}%
The rank $r^*$ therefore equals the number of identified shocks in Design~A, is smaller in Design~B, and is larger in Design~C.\approvedend

\approvedstart{29 October 2026}%
Each $h_k$ is scaled to have mean one. The generalised eigenvector estimator is the sample version of Proposition~\ref{prop:eigen} with $w=z+z^2/2$. It uses the two generalised eigenvectors of $(\hat M_w,\hat\Omega)$ whose eigenvalues are furthest from the sample mean of $w$. Table~\ref{tab:sim} reports four statistics over 1,000 replications for $n=250$, $1{,}000$ and $4{,}000$. The row angle between an estimated row $\hat a$ and a true row $a$ is
\begin{equation}\label{eq:rowangle}
\varphi=\arccos\frac{\lvert\hat a'a\rvert}{\lVert\hat a\rVert\,\lVert a\rVert},
\end{equation}
the angle whose cosine is the absolute uncentred correlation of the two rows, so it does not depend on the sign or scale of $\hat a$. The two estimated rows are matched to the two true rows of $A_1$, and the table reports the median over replications of the average of the two angles. The subspace angle compares the plane spanned by the two estimated rows with the plane spanned by the two true rows, so it ignores how the rows are rotated within that plane. It is the largest of the angles between the two planes, whose cosines are the singular values of $\hat U'U$ for orthonormal bases $\hat U$ and $U$ of the planes. When the two shocks form one class, as in Design~C, only the plane is identified, so the subspace angle can shrink while the row angle cannot. The fifth column concerns the rank of $C_\psi$, with $\psi=(z,z^2,z^3,d)$, which is the number of its nonzero singular values $s_1\ge s_2\ge s_3\ge s_4$, estimated by those of the sample matrix. Since $\rank C_\psi\le1$ exactly when $s_2=0$, the column reports the mean over replications of $\hat s_2/\hat s_1$, which tends to zero when $\rank C_\psi=1$ and to a positive number when it is 2. Its level depends on how $\psi$ is scaled, so only its behaviour as $n$ grows matters. The last column reports how often the rank test of \textcite{KleibergenPaap2006}, built from the singular values of a standardised $C_\psi$, rejects $\rank C_\psi\le1$ at the 5 percent level.\approvedend

\begin{table}[ht]
\centering
\caption{Estimation error and rank of $C_\psi$ in the simulation exercise.}\label{tab:sim}
\small
\setlength{\tabcolsep}{3pt}
\begin{tabular}{llcccc}
\toprule
Design & $n$ & Row angle & Subspace angle & $\hat s_2/\hat s_1$ & KP \\
\midrule
A (covariate) & 250 & 18.5 & 13.0 & 0.365 & 0.79 \\
 & 1,000 & 12.3 & 6.3 & 0.406 & 1.00 \\
 & 4,000 & 7.1 & 3.1 & 0.445 & 1.00 \\
\midrule
B (two regimes, distinct ratios) & 250 & 32.2 & 51.3 & 0.315 & 0.09 \\
 & 1,000 & 15.6 & 25.2 & 0.250 & 0.08 \\
 & 4,000 & 6.7 & 11.1 & 0.141 & 0.10 \\
\midrule
C (two regimes, equal ratios) & 250 & 24.6 & 21.4 & 0.259 & 0.08 \\
 & 1,000 & 19.3 & 9.5 & 0.147 & 0.07 \\
 & 4,000 & 18.2 & 4.7 & 0.075 & 0.07 \\
\bottomrule
\end{tabular}

\medskip
\parbox{\linewidth}{\footnotesize Notes: Angles are medians over the 1,000 replications, in degrees, and $\hat s_2/\hat s_1$ is a mean. KP is the rejection frequency of the \textcite{KleibergenPaap2006} test of $\rank C_\psi\le1$ at the 5 percent level. The replication code is part of the supplementary material.}
\end{table}

\approvedstart{29 October 2026}%
In Design~A both rows are estimated consistently, $\hat s_2/\hat s_1$ settles at a positive value, as $\rank C_\psi=r^*=2$ for this choice of $\psi$, and the Kleibergen--Paap test rejects $\rank C_\psi\le1$ in every replication from $n=1{,}000$. The subspace angle is much smaller than the row angle, because the two heteroskedastic eigenvalues, about $1.62$ and $1.33$, are close to each other relative to the noise in the weighted moments, so most of the error is a rotation within the pair. In Design~B both rows are again estimated consistently, although $\hat s_2/\hat s_1$ tends to zero and the test rejects $\rank C_\psi\le1$ in only about 9 percent of replications. The rank of $C_\psi$ is therefore one although two shocks are identified. The subspace angle is three to four times larger than in Design~A, because the eigenvalues of the two heteroskedastic shocks, about $0.90$ and $0.63$, are close to that of the homoskedastic ones, $0.50$. This is a form of the weak-identification problem studied by \textcite{Lewis2022}, and it affects the second row most. At $n=4{,}000$ the median angle is 3.0 degrees for the first row and 9.9 degrees for the second, whose eigenvalue is closer to $0.50$. In Design~C the individual rows are not estimated consistently, and their angle does not shrink towards zero, but the subspace they span is estimated consistently, as Lemma~\ref{lem:main}(c) predicts. The rank test gives the same answer as in Design~B, although here no shock is identified individually, so it cannot tell the two designs apart. The test is somewhat oversized at these sample sizes.\approvedend

\approvedstart{29 October 2026}%
The supplementary material contains the replication code and four exercises based on these designs, such as rerunning Design~A with $w=z$ to see the failure predicted in Section~\ref{sec:6}.\approvedend

\section{Concluding remarks}\label{sec:9}

\approvedstart{29 October 2026}%
In this note, we show that one lemma, proved with elementary linear algebra, organises the main results on identification through heteroskedasticity. Distinct variance paths are necessary and sufficient for identifying a shock, and the standard results on full and partial identification follow as special cases.\approvedend

\approvedstart{29 October 2026}%
Our key teaching point is to separate the statistical step, which differs across models, from the algebraic step, the identification problem in the classical sense, which is always the same. Counting moments and unknowns gives the order condition in each model, and the lemma gives the rank condition. Proposition~\ref{prop:eigen} links the lemma to the two-regime solution of \textcite{Rigobon:03}, and Proposition~\ref{prop:rank} explains why a rank-based count of linearly independent centred variance paths and a count of identified shocks can differ. The simulation illustrates both points. The generalised eigenvector estimator recovers both heteroskedastic shocks in a two-regime design in which the deviations of their variance paths from the mean are proportional, so that a rank test would count only one. When the two shocks share a variance path, it recovers only the space they span.\approvedend

\approvedstart{29 October 2026}%
Finally, we note some limitations and directions for further teaching material. The lemma concerns the algebraic step only. Estimation and inference require separate arguments. When the variance paths are distinct but nearly proportional, the shocks are only weakly identified, and inference that is robust to weak identification \parencite{Lewis2022} is needed. Using several weighted moments at once, by joint diagonalisation as in \textcite{PhamCardoso2001}, is a natural extension of Proposition~\ref{prop:eigen}.\approvedend

\printbibliography

\end{document}